\documentclass[sigconf]{acmart}

\newtheorem{theorem}{Theorem}
\newtheorem{assumption}{Assumption}
\usepackage{enumitem}
\usepackage{multirow}
\usepackage{graphicx}
\usepackage[normalem]{ulem}
\useunder{\uline}{\ul}{}

\AtBeginDocument{%
  }

\acmISBN{978-1-4503-XXXX-X/2018/06}

\begin{document}

\title{DrEM: Dual-Side Robust Ensemble Ranking from Noisy User Preference Predictions in Video Recommendation}

\author{Canwei Huang}
\authornote{Both authors contributed equally to this research.}
\email{2410103036@mails.szu.edu.cn}
\affiliation{%
  \institution{Shenzhen University}
  \city{Shenzhen}
  \country{China}
}

\author{Tiantian He}
\authornotemark[1]
\email{hetiantian05@kuaishou.com}
\affiliation{%
  \institution{Kuaishou Technology}
  \city{Beijing}
  \country{China}
}

\author{Xiaoxiao Xu}
\email{xuxiaoxiao05@kuaishou.com}
\affiliation{%
  \institution{Kuaishou Technology}
  \city{Beijing}
  \country{China}
}

\author{Jun Zhang}
\email{zhangjun08@kuaishou.com}
\affiliation{%
  \institution{Kuaishou Technology}
  \city{Beijing}
  \country{China}
}

\author{Ziran Deng}
\email{2510105033@mails.szu.edu.cn}
\affiliation{%
  \institution{Shenzhen University}
  \city{Shenzhen}
  \country{China}
}

\author{Weike Pan}
\email{panweike@szu.edu.cn}
\affiliation{%
  \institution{Shenzhen University}
  \city{Shenzhen}
  \country{China}
}

\author{Chunjie Chen}
\email{fengyu08@kuaishou.com}
\affiliation{%
  \institution{Kuaishou Technology}
  \city{Beijing}
  \country{China}
}

\author{Kaiqiao Zhan}
\email{zhankaiqiao@kuaishou.com}
\affiliation{%
  \institution{Kuaishou Technology}
  \city{Beijing}
  \country{China}
}

\renewcommand{\shortauthors}{Huang et al.}

\begin{abstract}
  Industrial video recommendation systems typically adopt a multi-stage architecture.
At the ensemble ranking stage, multi-dimensional user preference predictions (pxtrs) from an upstream multi-task model are fused into a unified ranking score to reflect user satisfaction.
Since users' true satisfaction is difficult to observe directly, ensemble ranking models commonly use pxtrs both as input features and as a source for constructing proxy preferences.
However, as outputs of an upstream prediction model, pxtrs inevitably contain prediction noise, which propagates to downstream learning across two sides.
On the supervision side, noisy pxtrs may flip proxy preferences and introduce erroneous gradients.
On the feature side, pxtr noise may propagate through model inputs and destabilize ranking scores.
Existing ensemble ranking methods typically treat pxtrs as reliable signals and overlook such prediction noise.
To address this, we propose DrEM, a \textbf{d}ual-side \textbf{r}obust \textbf{e}nse\textbf{m}ble ranking framework.
Our DrEM introduces a risk-denoising robust loss that corrects the empirical risk using estimated preference flip probability.
Meanwhile, it samples perturbations from the distribution of prediction noise and introduces a preference-preserving ranking consistency regularizer to improve feature-side output stability.
Theoretically, we obtain an approximate distribution of the prediction noise and prove that the robust loss remains superior under flip probability estimation error.
Extensive offline experiments and large-scale online A/B tests demonstrate the effectiveness and robustness of our DrEM.

\end{abstract}



\maketitle

\section{Introduction}
\label{sec:introduction}

In short-video recommendation, an ensemble ranking model ranks candidate items according to how well they match a user's preferences.
However, user satisfaction is difficult to observe because it is expressed through heterogeneous behaviors and cannot be fully characterized by a single ground-truth label.
Although posterior feedback (e.g., likes, comments, and follows) has clear behavioral semantics, it has two fundamental limitations.
Firstly, no actual feedback is available for unexposed items, whereas the ensemble ranking model must compare them all. 
Secondly, posterior signals are sparse and capture only specific aspects of user satisfaction, making them insufficient for multi-dimensional ranking.
Industrial systems therefore commonly adopt a two-stage pipeline.
An upstream multi-task model generates multi-dimensional predictions of a user's propensities (collectively referred to as pxtrs), including those for likes, comments, and others.
A downstream ensemble ranking model fuses these pxtrs into a unified ranking score. 
In this pipeline, pxtrs play a dual role: their relative values are used to construct proxy supervision, while the values themselves serve as input features for ranking inference~\cite{He2025EMER,Li2026EASQ,Shao2026UAME}.

Existing ensemble ranking methods have advanced the modeling capacity through learnable fusion weights~\cite{Li2023IntEL}, unsupervised objectives preserving pxtr ordinal and numerical information~\cite{Yu2024UREM}, and pairwise supervision derived from pxtr comparisons~\cite{He2025EMER,Li2026EASQ}.
However, they share the assumption that pxtrs are reliable signals, overlooking the prediction noise inherent in upstream model outputs~\cite{Kendall2018MTL}.
This noise propagates downstream to both sides, causing preference flips on the supervision side and output instability on the feature side.
Yet no existing method explicitly models the unique characteristics of this scenario.
The core challenge is the multi-dimensional heterogeneity of the prediction noise, which varies across pxtr tasks, across items, and consequently across pairs.
A uniform denoising strength would over-correct reliable pairs while under-correcting noisy ones, and a uniform perturbation scale would over-smooth confident items while leaving uncertain ones insufficiently protected.
Existing robustness methods cannot accommodate this heterogeneity for different reasons.
Denoising methods operate pointwise without modeling probabilities of preference flip~\cite{Chua2024UDT,He2024DCF}, while robust losses apply bounded functions uniformly across pairs without pair-level adaptation~\cite{Wu2024SSM,Yang2024PSL,Wu2024BSL,Yang2025SLK}, and perturbation methods either optimize loss smoothness rather than modeling prediction error~\cite{He2018APR,Zhang2024VAT} or use a global noise scale that ignores item-level heteroscedasticity~\cite{Ramazanli2024LSPR}.
Moreover, these methods target only one side, leaving the other side unprotected.

To address these limitations, both sides need to achieve sample-level differentiation while keeping their correction targets aligned.
On the supervision side, pair-level differentiation is essential because pairs with a small pxtr margin or large prediction noise are more likely to flip and require stronger correction, whereas the opposite holds for reliable pairs.
On the feature side, item-level differentiation is equally important because items with large prediction variance need stronger perturbations to improve robustness, while items with low variance should receive mild perturbations to avoid over-smoothing.
Furthermore, the two sides should share a consistent noise model so that their correction targets remain aligned.
Implementing this shared model requires estimating the noise distribution, which is feasible from data already available in the ensemble ranking scenario.
Observed user behaviors systematically deviate from the predicted pxtrs, and the magnitude of this deviation reflects the prediction noise.
By aggregating such deviations across large-scale data, the noise variance of each item can be estimated.
This estimation strategy satisfies both requirements, providing item-level and pair-level differentiation for each side and aligning their correction targets through shared parameters of the noise distribution.

Based on this analysis, we propose DrEM, a \textbf{d}ual-side \textbf{r}obust \textbf{e}nse\textbf{m}ble ranking framework.
On the supervision side, we derive approximate solutions to probabilities of preference flip based on the estimated noise distribution and construct a robust pairwise loss to correct the empirical risk toward the clean risk.
On the feature side, we sample perturbations from the distribution of prediction noise and introduce a preference-preserving ranking consistency regularizer to stabilize model outputs while avoiding conflict with the primary ranking objective.
Both components are derived from a shared noise model, aligning dual-side correction targets to the same noise source and synchronizing correction via shared parameters.

Our main contributions are summarized as follows:
\begin{itemize}[leftmargin=*]
\item We explicitly identify prediction noise arising from the dual role of pxtrs in industrial ensemble ranking for video recommendation, and systematically characterize how it propagates through both the supervision and feature sides.

\item We design a novel risk-denoising robust pairwise loss to correct the empirical risk by deriving approximate solutions to probabilities of preference flip.

\item We propose a novel preference-preserving ranking consistency regularizer using perturbations sampled from the distribution of prediction noise to stabilize model outputs.

\item We conduct extensive offline experiments on an industrial short-video recommendation platform, demonstrating the robustness of our DrEM across multiple tasks.
Online A/B tests further show significant gains on several business metrics.
\end{itemize}

\section{Related Work}
\label{sec:related_work}

\subsection{Ensemble Ranking for Short-Video Recommendation}

Early ensemble ranking approaches rely on manually designed heuristic formulas with RL-learned weights~\cite{Zhang2022BatchRLMTF}, motivating subsequent work to explore learnable fusion models~\cite{Cao2025xMTF,Cao2025Pantheon}.
IntEL uses behavioral labels as supervision to learn intent-aware, item-level personalized weights for pxtr fusion~\cite{Li2023IntEL}.
However, these labels are sparse and reflect only limited aspects of user satisfaction, providing insufficient information for complex fusion modeling.
UREM addresses this limitation with unsupervised learning objectives that preserve the ordinal and numerical information in pxtrs~\cite{Yu2024UREM}.
Recently, EMER derives pairwise supervision directly from pxtr comparisons, models interactions among candidate items with a Transformer~\cite{Vaswani2017Attention}, and uses a tailored loss to address the lack of a unified satisfaction label~\cite{He2025EMER}.
Turning to supervision sources, EASQ incorporates sparse questionnaire signals to directly align ranking with user satisfaction, complementing behavioral signals~\cite{Li2026EASQ}.
These methods improve ensemble ranking through advances in fusion mechanisms, supervision design, and signal sources, but leave pxtr prediction noise unmodeled.

\subsection{Learning with Noisy Labels}

Research on ranking under imperfect supervision has evolved along several lines.
Existing studies identify noisy implicit-feedback interactions from their losses and reduce their influence through reweighting, truncation, correction, or resampling~\cite{Wang2021ADT,He2024DCF,Zhang2025PLD}.
UDT later theoretically establishes a strict positive correlation between the relative magnitude of cross-entropy losses and both preference uncertainty and user inconsistency~\cite{Chua2024UDT}.
Another line focuses on robust loss design.
SSM shows that the sampled softmax loss mitigates popularity bias and facilitates hard-negative mining~\cite{Wu2024SSM}.
PSL attributes the softmax loss's false-negative sensitivity to its exponential function and replaces it with a bounded alternative, yielding an objective equivalent to robust optimization over the negative-sample distribution~\cite{Yang2024PSL}.
Other approaches use early memorization or cross-model agreement as denoising signals~\cite{Gao2022SGDL,Wang2022DeCA}, model pointwise label-flip probabilities to distinguish true negatives from mislabeled positives~\cite{Yu2020NBPO}, balance hard and false negatives with difficulty-aware or adversarial contrastive objectives~\cite{Yang2022HDCCF,Zhang2023AdvInfoNCE}, or reformulate pairwise learning as variational inference to unify preference alignment and denoising~\cite{Liu2026VarBPR}.
However, these methods target noise in observed posterior feedback or uncertainty in sampled negatives rather than pxtr prediction noise, and do not explicitly model pairwise preference flip probabilities.

\subsection{Perturbation-Based Regularization}

Feature-side robustness concerns model stability under input perturbations.
Adversarial training is a standard approach that trains models against worst-case perturbations through minimax optimization.
APR establishes this paradigm with globally fixed-magnitude perturbations~\cite{He2018APR}.
However, a uniform magnitude ignores differences in user vulnerability and may either hurt robust users or underprotect vulnerable ones.
VAT therefore adapts perturbation magnitude to user vulnerability~\cite{Zhang2024VAT}, while CascadeAT scales perturbations according to position-dependent cascade effects in sequential recommendation~\cite{Tan2024CascadeAT}.
Beyond magnitude, existing studies explore perturbation objectives, efficiency, and injection locations.
GaussAug generates perturbations via both embedding-centered and gradient-guided Gaussian sampling, then selects the one with the largest loss~\cite{Wang2023Gaussian}.
SharpCF uses a trajectory loss based on historical model states to reduce minimax overhead~\cite{Chen2023SharpCF}.
LSPR applies dropout to sparse and Gaussian noise to dense features with loss-balanced weighting~\cite{Ramazanli2024LSPR}, AMR perturbs deep image feature vectors~\cite{Tang2020AMR}, and RAT injects noise into hidden-layer outputs~\cite{Chen2024RAT}.
However, these methods either enforce loss smoothness rather than model prediction noise or use a global perturbation scale, and even some require historical training states unsuited to streaming training.

\section{Methodology}
\label{sec:method}

In this section, we study the ensemble ranking problem under noisy preference prediction.
We first formalize the problem and the notations (Section~\ref{sec:preliminaries}).
Built on a shared logit-space noise model, we then correct the supervision flip on the supervision side (Section~\ref{sec:label_side}) and keep the model output stable against input noise on the feature side (Section~\ref{sec:feature_side}).
Finally, we present the overall training objective (Section~\ref{sec:objective}).
Figure~\ref{fig:main} provides an overview of the proposed method.

\begin{figure*}[t]
  \centering
  \includegraphics[width=\textwidth]{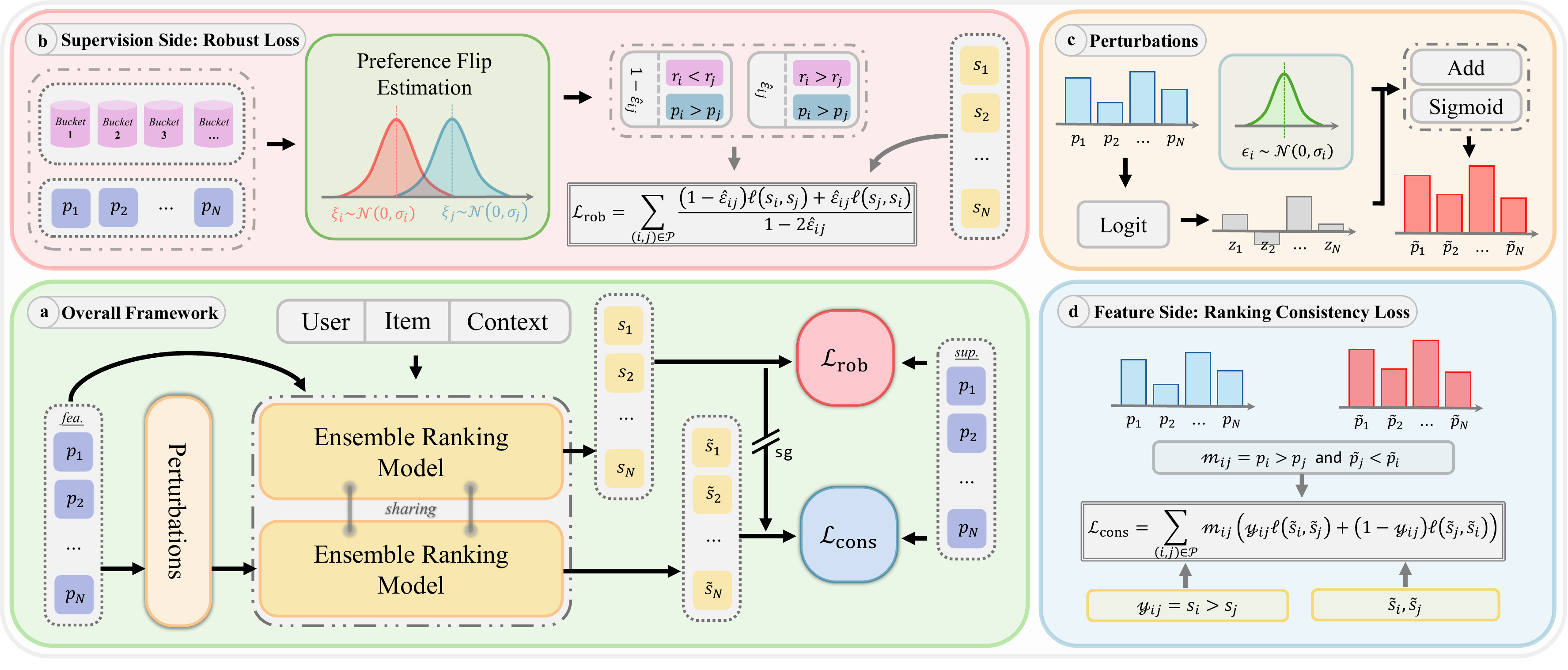}
  \caption{
    Overview of our DrEM.
    (a) The overall framework: pxtrs from the upstream multi-task model serve both as proxy supervision and as input features, and their prediction noise is corrected on both sides under a shared logit-space noise model.
    (b) The supervision side (Section~\ref{sec:label_side}), a risk-denoising robust loss that corrects the empirical risk pair by pair using the estimated flip probability $\hat\varepsilon_{ij}$.
    (c) The perturbation scheme (Section~\ref{sec:feature_side}): noise is added to the pxtr logits, and preference-preserving filtering keeps only the pairs whose order is unchanged.
    (d) The feature side (Section~\ref{sec:feature_side}), a ranking consistency regularizer that constrains the perturbed scores to match the original order on the retained pairs.}
  \label{fig:main}
\end{figure*}

\subsection{Preliminaries}
\label{sec:preliminaries}

Consider the candidate item set $\mathcal{I}$ of a single request.
For each candidate item $i \in \mathcal{I}$, the upstream multi-task model outputs predictions over multiple dimensions, collectively denoted as pxtrs.
Let $\mathcal{X}$ be the set of pxtr types, that is, the user behaviors predicted by the multi-task model.
For example, $\mathcal{X}=\{\mathrm{ltr}, \mathrm{ftr}, \mathrm{cmtr}, \dots\}$ corresponds to the predicted propensities of behaviors of liking, following, and commenting, etc.
We organize the pxtrs of all candidate items into a matrix $\mathbf{P}\in\mathbb{R}^{|\mathcal{I}|\times|\mathcal{X}|}$, where the entry $p_{i,x}\in[0,1]$ denotes the predicted value of item $i$ on pxtr type $x$.
For instance, $p_{i,\mathrm{ltr}}$ is the predicted like rate of item $i$.
Correspondingly, let $\mathbf{R}\in\mathbb{R}^{|\mathcal{I}|\times|\mathcal{X}|}$ be the latent clean pxtr matrix, whose entry $r_{i,x}$ is the latent noise-free value of $p_{i,x}$.

The ensemble ranking model $\mathrm{EM}_\theta(\cdot)$ takes three types of features as input.
These are the pxtr matrix $\mathbf{P}$ of the candidate items, user-side features such as user profiles and historical behavior sequences, and context features such as request scenario and time.
The goal of the ensemble ranking model is to make the output ranking scores $\mathbf{s}=(s_1,\dots,s_{|\mathcal{I}|})^\top\in\mathbb{R}^{|\mathcal{I}|}$ approximate the ranking induced by the true user satisfaction as closely as possible.
We define the preference pair set $\mathcal{S} = \{(i,j) \in \mathcal{I} \times \mathcal{I} \mid \mathrm{sat}_i > \mathrm{sat}_j\}$, where $\mathrm{sat}_i$ denotes the latent true satisfaction of the user toward item $i$.
The optimization objective is then
\begin{equation}
\label{eq:objective}
\theta^{*} = \arg\max_\theta \; \frac{1}{|\mathcal{S}|}\sum_{(i,j) \in \mathcal{S}} \delta(s_i, s_j), \quad \delta(s_i, s_j)=\begin{cases}1, & s_i > s_j,\\ 0, & s_i \le s_j,\end{cases}
\end{equation}
where $\delta(s_i,s_j)$ is a pairwise consistency indicator.

Since user satisfaction $\mathrm{sat}_i$ is not directly observable and lacks a unified label, pxtrs, as the predictions over multi-dimensional user behaviors, are used as proxy supervision signals.
For each pxtr type $x\in\mathcal{X}$, we construct a set of preference pairs from the relative order of the corresponding pxtrs, $\mathcal{P}_x=\{(i,j) \in \mathcal{I} \times \mathcal{I} \mid p_{i,x} > p_{j,x}\}$, which replaces $\mathcal{S}$ to guide the model toward learning the relative order of the ranking scores.
We carry out the subsequent derivation in logit space, where the prediction noise admits an additive Gaussian form under Assumption~\ref{asm:noise} in Section~\ref{sec:theory}.
Let $z_{i,x}=\mathrm{logit}(p_{i,x})$ and $z_{i,x}^*=\mathrm{logit}(r_{i,x})$ be the noisy and clean logit scores, collected into matrices $\mathbf{Z},\mathbf{Z}^*\in\mathbb{R}^{|\mathcal{I}|\times|\mathcal{X}|}$.
For brevity, the derivation below drops the type subscript $x$ and uses column vectors $\mathbf{z}=(z_1,\dots,z_{|\mathcal{I}|})^\top$ and $\mathbf{p}=(p_1,\dots,p_{|\mathcal{I}|})^\top$ for a single pxtr type.

\subsection{Supervision Side: A Risk-Denoising Robust Loss}
\label{sec:label_side}

Our method follows the noise-aware risk correction principle~\cite{Natarajan2013Learning,RayChowdhury2024Provably}.
When the observed supervision signal is a noisy version of a latent clean signal, the empirical risk should be corrected according to the corresponding noise mechanism.
In the scenario studied in this paper, the supervision signal is a pairwise preference induced by the upstream output pxtrs.
The noise mechanism therefore naturally manifests as a flip of the pairwise preference.
We accordingly instantiate the general principle above as a risk-denoising robust pairwise loss.

Let $(i^*,j^*)$ be a clean preference pair induced by the clean pxtr, that is, $r_{i^*}>r_{j^*}$.
Since the supervision used in training comes from the noisy predictions rather than clean pxtrs, the observed preference pair $(i,j)$ may flip relative to the clean pair.
It remains $(i^*,j^*)$ with probability $1-\varepsilon_{ij}$ and reverses to $(j^*,i^*)$ with flip probability $\varepsilon_{ij}$ (Theorem~\ref{thm:flip_prob} in Section~\ref{sec:theory}).
Taking the expectation over the flip randomness and the data distribution, the risk of the basic pairwise loss is
\begin{equation}
\label{eq:risk}
\mathcal{R} = \mathbb{E}_{(i^*,j^*)}\Big[(1-\varepsilon_{ij})\,\ell(s_{i^*}, s_{j^*}) + \varepsilon_{ij}\,\ell(s_{j^*}, s_{i^*})\Big].
\end{equation}
The extra reverse loss term $\varepsilon_{ij}\,\ell(s_{j^*}, s_{i^*})$ grows with the flip probability $\varepsilon_{ij}$ and drives the optimization gradient to deviate systematically from the clean objective.
This is the mechanism that the pxtr prediction noise may harm the ranking quality on the supervision side.

The weight of the reverse loss term is thus determined by the flip probability.
Intuitively, the flip probability depends on the logit margin and the noise distribution: a larger margin or a smaller noise yields a lower flip probability, and consequently the flip risk varies significantly across preference pairs.
Therefore, the correction should be performed pair by pair, relying on an estimate of the flip probability $\hat\varepsilon_{ij}$ of each pair $(i,j)$, rather than applying a globally uniform denoising strength.
The flip probability is computable from the noise distribution (Theorem~\ref{thm:flip_prob} in Section~\ref{sec:theory}), whose parameters can be estimated from large-scale posterior user feedback (Theorem~\ref{thm:variance} in Section~\ref{sec:theory}).

To correct the reverse loss term in the risk shown in Eq.~(\ref{eq:risk}), we derive the robust pairwise loss, which takes the form of a weighted combination of the forward and reverse pairwise losses,
\begin{equation}
\label{eq:rob_loss}
\mathcal{L}_{\text{rob}} = \sum_{(i,j)\in\mathcal{P}} \frac{(1 - \hat{\varepsilon}_{ij})\, \ell(s_i, s_j) - \hat{\varepsilon}_{ij}\, \ell(s_j, s_i)}{1 - 2\hat{\varepsilon}_{ij}},
\end{equation}
when $\hat{\varepsilon}_{ij}=\varepsilon_{ij}$, the risk of this loss equals the clean risk exactly.
In practice the flip probability can only be estimated rather than known exactly.
As long as $\hat{\varepsilon}_{ij}\in(0,1/2)$, the risk of the robust loss stays closer to the clean risk than that of the basic pairwise loss, and a more accurate estimate brings it even closer (Theorem~\ref{thm:superiority} in Section~\ref{sec:theory}). This property guarantees that the method remains effective even with imperfect flip probability estimates.

\subsection{Feature Side: A Preference-Preserving Ranking Consistency Regularizer}
\label{sec:feature_side}

The supervision-side correction in Section~\ref{sec:label_side} mitigates the effect of pxtr prediction noise on the supervision signal.
However, pxtrs also serve as input features to the ranking model, and their prediction noise may still cause unstable model outputs.
To keep the model output stable against such noise, we explicitly simulate the prediction noise that may exist in the pxtr inputs and constrain the model output accordingly.

For the pxtrs of a whole request input, we sample a perturbation vector $\boldsymbol{\epsilon} = (\epsilon_1,\dots,\epsilon_{|\mathcal{I}|})^\top$ from the distribution of prediction noise, where $\epsilon_i \sim \mathcal{N}(0,\sigma_i^2)$.
In logit space we construct the perturbed logit scores $\tilde{\mathbf{z}} = \mathbf{z} + \boldsymbol{\epsilon}$, where $\mathbf{z}=(z_1,\dots,z_{|\mathcal{I}|})^\top$ is the original logit score vector.
The sign of each component of $\boldsymbol{\epsilon}$ determines the perturbation direction, and its absolute value controls the perturbation strength.
The perturbed pxtrs are mapped back to probability space to obtain $\tilde{\mathbf{p}} = \sigma(\tilde{\mathbf{z}})$, where $\sigma(\cdot)$ is the element-wise sigmoid function.
This perturbation adopts the same logit-space noise model, with parameters $\sigma_i^2$ (Theorem~\ref{thm:variance} in Section~\ref{sec:theory}), as the flip probability derivation on the supervision side.

We construct a perturbed variant of the request input by replacing the original pxtrs with $\tilde{\mathbf{p}}$ while keeping the user-side and context features unchanged.
Given the original input and the perturbed input, the ensemble ranking model outputs the original ranking scores $\mathbf{s}\in\mathbb{R}^{|\mathcal{I}|}$ and the perturbed ranking scores $\tilde{\mathbf{s}}\in\mathbb{R}^{|\mathcal{I}|}$, respectively.

However, pxtrs themselves carry a strong ranking signal, and the perturbation may change the pairwise preference they express.
Suppose a pairwise order flips before and after the perturbation.
Forcing the perturbed scores to keep the original order would then make the model ignore the valid preference change in the perturbed features, and would conflict with the primary ranking objective.
To avoid this, we introduce a preference-preserving filtering mechanism.
It applies the ranking consistency constraint only to pairs whose preference is consistent before and after the perturbation.
Based on the pairwise preference indicator $\delta(\cdot,\cdot)$ defined in Section~\ref{eq:objective}, we denote the preferences before and after the perturbation as $\delta(p_i,p_j)$ and $\delta(\tilde{p}_i,\tilde{p}_j)$, and define the valid pair indicator as
\begin{equation}
m_{ij} = 1 - \big(\delta(p_i,p_j) \oplus \delta(\tilde{p}_i,\tilde{p}_j)\big),
\end{equation}
where $\oplus$ is the exclusive-or operation.
When the preference is consistent before and after perturbation, we have $m_{ij}=1$, and $m_{ij}=0$ otherwise.

Based on preference-preserving filtering, we define the pairwise ranking consistency loss as
\begin{equation}
\label{eq:cons_loss}
\mathcal{L}_{\text{cons}} = \sum_{(i,j)\in \mathcal{P}} m_{ij}\left( y_{ij}\,\ell(\tilde{s}_i, \tilde{s}_j) + (1 - y_{ij})\,\ell(\tilde{s}_j, \tilde{s}_i) \right),
\end{equation}
where $y_{ij}=\delta(s_i,s_j)$ is the relative preference expressed by the ranking scores under the original input, and $\ell(\cdot,\cdot)$ is the basic pairwise loss.
This loss takes the ranking scores under the original input as an anchor, and constrains the perturbed scores to keep a consistent order only on preference-preserving pairs.
It thus resists pxtr noise on the feature side while avoiding conflict with the primary ranking objective.

\subsection{Overall Optimization Objective}
\label{sec:objective}

Combining the risk-denoising robust loss in Eq.~(\ref{eq:rob_loss}) and the preference-preserving ranking consistency regularizer in Eq.~(\ref{eq:cons_loss}), the overall optimization objective is
\begin{equation}
\label{eq:overall_objective}
\mathcal{L} = \mathcal{L}_{\text{rob}} + \lambda_{\text{cons}}\mathcal{L}_{\text{cons}},
\end{equation}
where $\mathcal{L}_{\text{rob}}$ mitigates the proxy-pair flips caused by pxtr prediction noise from the supervision side, and $\mathcal{L}_{\text{cons}}$ improves the output stability of the model against pxtr prediction noise from the feature side.
Both share the same logit-space noise model and parameters, and jointly eliminate the effect of prediction noise on the supervision side and the feature side.

Compared with the basic pairwise loss, the extra cost of our DrEM mainly comes from the perturbed forward pass.
Other operations, such as the computation of preference flip probability and preference-preserving filtering, are scalar operations with negligible cost, and there is no additional cost at inference time.
Although our DrEM increases the training-time computation cost, training can be performed asynchronously and is decoupled from the latency-critical inference pipeline.

\section{Theoretical Analysis}
\label{sec:theory}

In this section, we provide theoretical foundations for the supervision-side robust loss and the feature-side perturbation design.
Both sides rely on a concrete characterization of prediction noise to yield tractable results.
We model the prediction noise as additive zero-mean Gaussian in logit space, which preserves probability boundedness, aligns with asymptotic normality for parametric models and approximately normal residuals for complex models~\cite{Zhigalskii2026Uncertainty}, and yields closed-form solutions via probit approximation.

\begin{assumption}[Additive Gaussian noise model]
\label{asm:noise}
For each item $i$, let $z_i = \mathrm{logit}(p_i)$ and $z_i^* = \mathrm{logit}(r_i)$ be the logits of the noisy and clean pxtrs, respectively.
They are related by
\begin{equation}
z_i = z_i^* + \xi_i,
\end{equation}
where $\xi_i \sim \mathcal{N}(0, \sigma_i^2)$ and $\xi_i \perp \xi_j$ for $i \neq j$.
\end{assumption}

Under the additive Gaussian noise model in Assumption~\ref{asm:noise}, the preference flip probability $\varepsilon_{ij}$ of a proxy pair $(i,j)$ cannot be observed directly and can only be approximately estimated as $\hat\varepsilon_{ij}$.
The estimation error prevents the robust loss from recovering the clean loss exactly, and its denoising effect varies with the estimation accuracy.
For a given $\hat\varepsilon_{ij}$, we characterize how much the robust loss improves over the basic pairwise loss.

\begin{theorem}[Risk superiority of the robust pairwise loss]
\label{thm:superiority}
Let the flip probability $\varepsilon_{ij}\in(0,1/2)$ and its estimate $\hat\varepsilon_{ij}\in(0,1/2)$.
Then the denoising effect of the robust pairwise loss improves monotonically as $|\hat\varepsilon_{ij}-\varepsilon_{ij}|$ decreases, and is strictly stronger than that of the basic pairwise loss.
When $\hat\varepsilon_{ij}=\varepsilon_{ij}$, the expected robust loss equals the clean loss exactly.
\end{theorem}

\begin{proof}
Consider a single clean preference pair $(i^*,j^*)$ with $r_{i^*}>r_{j^*}$.
The observed pair $(i,j)$ flips to $(j^*,i^*)$ with probability $\varepsilon_{ij}$ and remains $(i^*,j^*)$ with probability $1-\varepsilon_{ij}$, so the expectation of a pairwise loss $\ell$ is
\begin{equation}
\mathbb{E}_{(i^*,j^*)}[\ell] = (1-\varepsilon_{ij})\,\ell(s_{i^*}, s_{j^*}) + \varepsilon_{ij}\,\ell(s_{j^*}, s_{i^*}).
\end{equation}
Substituting the definition of the robust loss, the coefficient of the clean term $\ell(s_{i^*}, s_{j^*})$ in $\mathbb{E}_{(i^*,j^*)}[\mathcal{L}_{\text{rob}}]$ is
\begin{equation}
c(\hat\varepsilon_{ij}) = \frac{1-\varepsilon_{ij}-\hat\varepsilon_{ij}}{1-2\hat\varepsilon_{ij}},
\end{equation}
and the coefficient of the reverse term $\ell(s_{j^*}, s_{i^*})$ is $1-c(\hat\varepsilon_{ij})$.
In the expectation of either loss, the two coefficients sum to 1, so the denoising effect is directly determined by $c(\hat\varepsilon_{ij})$, and the closer $c(\hat\varepsilon_{ij})$ is to 1, the stronger the denoising effect.
The basic pairwise loss corresponds to $\hat\varepsilon_{ij}=0$, that is, $c(0)=1-\varepsilon_{ij}$.

Since $c(\hat\varepsilon_{ij})-c(0)>0$ holds for all $\varepsilon_{ij},\hat\varepsilon_{ij}\in(0,1/2)$, the clean-term coefficient of the robust loss is strictly larger than that of the basic pairwise loss, and its denoising effect is stronger.
Moreover, since $c'(\hat\varepsilon_{ij})>0$, $c(\hat\varepsilon_{ij})$ is strictly increasing in $\hat\varepsilon_{ij}$ on $(0,1/2)$ with $c(\varepsilon_{ij})=1$.
Hence the closer $\hat\varepsilon_{ij}$ is to $\varepsilon_{ij}$, the closer $c(\hat\varepsilon_{ij})$ is to 1 and the stronger the denoising effect.
When $\hat\varepsilon_{ij}=\varepsilon_{ij}$, we have $c(\hat\varepsilon_{ij})=1$, the reverse-term coefficient is zero, and $\mathbb{E}_{(i^*,j^*)}[\mathcal{L}_{\text{rob}}]=\ell(s_{i^*}, s_{j^*})$ is exactly the clean loss.
\end{proof}

The denoising effect of the robust loss depends on the accuracy of $\hat\varepsilon_{ij}$.
The problem thus reduces to estimating the flip probability accurately.
Under the additive logit noise model, the flip probability admits an approximate solution derived from the parameters of noise distribution (Theorem~\ref{thm:flip_prob}).

\begin{theorem}[Approximate solution to the preference flip probability]
\label{thm:flip_prob}
Under Assumption~\ref{asm:noise}, the preference flip probability of a proxy pair $(i, j)\in\mathcal{P}$ is estimated by
\begin{equation}
\hat{\varepsilon}_{ij} = \Phi\left(-\frac{z_i - z_j}{\sqrt{\sigma_i^2 + \sigma_j^2}}\right),
\end{equation}
where $\Phi(\cdot)$ is the cumulative distribution function of the standard normal distribution.
\end{theorem}

\begin{proof}
Under Assumption~\ref{asm:noise}, the logit margin $d_{ij} := z_i - z_j$ follows a normal distribution with location parameter $\mu_{ij}$ and variance $\sigma_i^2 + \sigma_j^2$, where $\sigma_i^2$ and $\sigma_j^2$ are the known estimates from Theorem~\ref{thm:variance}.
The proxy pair $(i,j)$ flips if and only if the noise reverses the sign of $d_{ij}$, so the flip probability $\varepsilon_{ij}$ is the standard normal tail probability of the standardized $\mu_{ij}$, that is, $\varepsilon_{ij}$ is a monotone measurable function of $\mu_{ij}$.

Given the observation $d_{ij}$, the log-likelihood of $\mu_{ij}$ is
\begin{equation}
\ell(\mu_{ij}) = -\frac{(d_{ij} - \mu_{ij})^2}{2(\sigma_i^2 + \sigma_j^2)} + C.
\end{equation}
By the factorization theorem, $d_{ij}$ is a sufficient statistic for $\mu_{ij}$.
Taking the derivative of $\ell$ with respect to $\mu_{ij}$ and setting it to zero gives the solution $\hat\mu_{ij} = d_{ij}$, and the second derivative $-(\sigma_i^2+\sigma_j^2)^{-1} < 0$ confirms that it is a maximum, so $d_{ij}$ is the MLE of $\mu_{ij}$.

By the invariance property of the MLE, the MLE of any measurable function of a parameter equals the value of that function at the MLE of the parameter.
As a measurable function of $\mu_{ij}$, the MLE of $\varepsilon_{ij}$ is the value of that function at $\hat\mu_{ij} = d_{ij}$.
Expanding it yields the conclusion.
\end{proof}

The approximate flip probability in Theorem~\ref{thm:flip_prob} still depends on the noise variance $\sigma_i^2$, which cannot be obtained from a single observation.
However, the convolution of the sigmoid and the Gaussian noise produces a systematic difference between the posterior behavior and the predicted rate.
This difference encodes the information of the variance and can be extracted through a bucketing aggregation strategy (Theorem~\ref{thm:variance}).

\begin{theorem}[Probit estimation of the logit-space noise variance]
\label{thm:variance}
Under Assumption~\ref{asm:noise}, suppose all items in bucket $B_k$ share the same logit-space noise variance $\sigma_k^2$, and let $\mathrm{diam}(B_k) = \max_{i,j \in B_k}|z_i - z_j|$.
Then when $|B_k| \to \infty$ and $\mathrm{diam}(B_k) \to 0$, we have
\begin{equation}
\hat{\sigma}_k^2 = \frac{1}{\lambda^2}\left[\left(\frac{\Phi^{-1}(\mathbb{E}_{B_k}[p_i])}{\Phi^{-1}(\mathbb{E}_{B_k}[y_i])}\right)^2 - 1\right] \rightarrow \sigma_k^2, \quad \lambda = \sqrt{\pi/8}.
\end{equation}
\end{theorem}

\begin{proof}
By definition, $\mathbb{E}[r_i] = \mathbb{E}[\sigma(z_i + \xi_i)]$.
When $\sigma_i^2 = 0$, we have $\mathbb{E}[r_i] = p_i$.
When $\sigma_i^2 > 0$, by Jensen's inequality, $\mathbb{E}[r_i] \neq p_i$.
By the probit approximation of the sigmoid $\sigma(x) \approx \Phi(\lambda x)$,
\begin{equation}
\mathbb{E}[r_i] \approx  \Phi\left(\frac{\lambda z_i}{\sqrt{1 + \lambda^2\sigma_i^2}}\right) = \Phi\left(\frac{\Phi^{-1}(p_i)}{\sqrt{1 + \lambda^2 \sigma_i^2}}\right).
\end{equation}
Since $\mathbb{E}[y_i \mid z_i]$ is the conditional expectation of $y_i$, it cannot be estimated from a single observation and needs to be aggregated within a bucket.
Given $z_i$, the $y_i$ are independent $\text{Bernoulli}(r_i)$, so by the law of large numbers, when $|B_k| \to \infty$ we have $\mathbb{E}_{B_k}[y_i] =\mathbb{E}_{B_k}[r_i]$.
When $\mathrm{diam}(B_k) \to 0$, the $p_i$ within the bucket become uniform, and
\begin{equation}
\mathbb{E}_{B_k}[y_k] =\Phi\left(\frac{\Phi^{-1}(p_k)}{\sqrt{1 + \lambda^2 \sigma_k^2}}\right).
\end{equation}
Solving for $\sigma_k^2$ yields the result.
\end{proof}

\section{Experiments}
\label{sec:experiments}

In this section, we conduct extensive offline and online experiments to evaluate the effectiveness of our DrEM and answer the following research questions: (1) can our DrEM maintain superior ranking performance over the existing methods under instability in the input pxtr features, (2) how do the contributions of the supervision-side correction and feature-side regularizer vary across different perturbation strengths, (3) can our DrEM improve business metrics in online A/B tests, and (4) how does the distribution of pairwise preference flip probabilities explain the effectiveness of our method.

\subsection{Experimental Setup}
\label{sec:experimental_setup}

\subsubsection{Dataset}
\label{sec:dataset}

All offline experiments are conducted on real-world data from a large-scale industrial short-video social platform with hundreds of millions of DAUs, where the upstream multi-task model predicts pxtrs for multiple user behaviors, such as clicks, views, long views, likes, follows, comments, and forwards.
Note that public short-video datasets~\cite{Ni2025MicroLens,Gao2022KuaiRec,Sun2023KuaiSAR} only contain posterior feedback labels without upstream pxtr outputs, which thus cannot be used since our DrEM exploits pxtrs both as pairwise supervision signals and input features.

To evaluate the robustness of our DrEM to the pxtr prediction noise, we inject Gaussian perturbations into the input pxtrs to construct evaluation sets with different perturbation strengths.
Specifically, for each pxtr, the perturbed logit score is computed as
$\tilde{z}_i = z_i + \alpha \epsilon_i,$
where $\epsilon_i$ denotes a sampled Gaussian perturbation and $\alpha$ controls the perturbation strength.
We set $\alpha$ to $0.5$, $1.0$, and $2.0$, corresponding to mild, moderate, and severe pxtr prediction noise, respectively.

\subsubsection{Evaluation Metric}
\label{sec:evaluation_metric}

We use Group AUC (GAUC) as the primary offline metric.
In industrial ensemble ranking, the model output is a single scalar score per item and the final ranking is determined by sorting these scores.
Ranking consistency with the reference order is therefore the central criterion for model quality~\cite{He2025EMER,Cao2025Pantheon,Li2026EASQ}.
GAUC further improves upon global AUC by computing per-user AUC and weighting by impression count, thereby eliminating cross-user confounding and better reflecting personalized ranking performance.
It is computed as follows:
\begin{equation}
\mathrm{GAUC}
=
\frac{
\sum_{i}^{\#\mathrm{user}}
\#\mathrm{impression}_{u_i}\cdot\mathrm{AUC}_{u_i}
}{
\sum_{i}\#\mathrm{impression}_{u_i}
},
\end{equation}
where $u_i$ denotes the $i$-th user, $\#\mathrm{impression}_{u_i}$ is the number of impressions, and $\mathrm{AUC}_{u_i}$ is the corresponding AUC value.

On each perturbed evaluation set, we compute GAUC using the unperturbed pxtr ordering as the reference for the model ranking scores.
A smaller drop in GAUC as the perturbation strength increases indicates greater robustness to pxtr prediction noise.

\subsubsection{Baselines}
\label{sec:baselines}

We use two state-of-the-art ensemble ranking models, EMER~\cite{He2025EMER} and EASQ~\cite{Li2026EASQ}, as backbones, and compare with four robust learning baselines: SSM~\cite{Wu2024SSM} and PSL~\cite{Yang2024PSL} on the supervision side, and GaussAug (GA)~\cite{Wang2023Gaussian} and LSPR~\cite{Ramazanli2024LSPR} on the feature side (see Appendix~\ref{app:baselines} for details).
These baselines operate on either supervision-side losses or feature-side perturbations, motivating joint intervention under a shared noise model.
For ablation, we also evaluate DrEM:S and DrEM:F, which represent the supervision-side-only and feature-side-only variants of our method, respectively.

\subsubsection{Implementation Details}
\label{sec:implementation_details}

We use the same settings across different methods for a fair comparison.
Following~\cite{He2025EMER}, we use embedding dimensions of 32 for user and item features and 8 for pxtrs, and train the model with the Adam optimizer~\cite{Kingma2015Adam} at a learning rate of $5\times10^{-6}$.
We tune the consistency weight $\lambda_{\mathrm{cons}}$ in Eq.~(\ref{eq:overall_objective}) over $\{0.1,0.3,0.6\}$ to prevent the consistency regularizer from overwhelming the main ranking loss.

\subsection{Overall Performance (RQ1)}
\label{sec:overall_performance}

To answer RQ1, we compare all methods on both the EMER and EASQ backbones at a fixed perturbation strength $\alpha=1.0$.
Table~\ref{tab:overall} reports the GAUC of each method on several pxtr tasks (see Appendix~\ref{app:task_details} for task details), where we have the following observations:

\begin{table}[t]
\centering
\caption{Overall comparison at $\alpha\!=\!1.0$. Best in \textbf{red}, second best \underline{blue}. All GAUC values share the prefix ``0.''.}
\label{tab:overall}
\setlength{\tabcolsep}{4pt}%
\begin{tabular}{c|ccccccc}
\toprule
Backbone & \multicolumn{7}{c}{EMER} \\
\cmidrule(lr){1-1} \cmidrule(lr){2-8}
Task & pctr & pvtr & plvtr & pltr & pwtr & pcmtr & pftr \\
\midrule
Base   & .6530  & .6899 & .6851 & .6622 & .6775 & .6235 & .6853 \\
GA     & .6506 & .6930  & .6833 & .6630  & .6787 & .6227 & .6876 \\
LSPR   & .6549 & .6918 & .6863 & .6646 & .6789 & .6257 & .6860  \\
PSL    & .6499 & .6895 & .6830  & {\color[HTML]{0070C0} \underline{.6673}} & .6833 & {\color[HTML]{0070C0} \underline{.6287}} & .6886 \\
SSM    & .6516 & .6882 & .6832 & .6633 & .6781 & .6224 & .6853 \\
DrEM:F & {\color[HTML]{FF0000} \textbf{.6570}} & {\color[HTML]{0070C0} \underline{.6947}} & {\color[HTML]{FF0000} \textbf{.6877}} & .6653 & .6814 & .6260  & .6870  \\
DrEM:S & .6531 & .6934 & .6863 & .6662 & {\color[HTML]{FF0000} \textbf{.6936}} & .6286 & {\color[HTML]{0070C0} \underline{.6931}} \\
DrEM   & {\color[HTML]{0070C0} \underline{.6563}} & {\color[HTML]{FF0000} \textbf{.6966}} & {\color[HTML]{0070C0} \underline{.6876}} & {\color[HTML]{FF0000} \textbf{.6725}} & {\color[HTML]{0070C0} \underline{.6871}} & {\color[HTML]{FF0000} \textbf{.6310}} & {\color[HTML]{FF0000} \textbf{.6965}} \\
\midrule
Backbone & \multicolumn{7}{c}{EASQ} \\
\cmidrule(lr){1-1} \cmidrule(lr){2-8}
Task & pctr & pvtr & plvtr & pltr & pwtr & pcmtr & pftr \\
\midrule
Base   & .6352 & .6615 & .6528 & .6693 & .6811 & .6621 & .6804 \\
GA     & .6347 & .6622 & .6519 & .6700 & .6815 & .6633 & .6807 \\
LSPR   & .6359 & .6619 & .6538 & .6704 & .6822 & .6630 & .6818 \\
PSL    & .6348 & .6639 & .6531 & .6711 & .6826 & .6640 & .6828 \\
SSM    & .6353 & .6663 & .6546 & {\color[HTML]{0070C0} \underline{.6741}} & .6833 & .6644 & {\color[HTML]{0070C0} \underline{.6850}} \\
DrEM:F & {\color[HTML]{0070C0} \underline{.6367}} & .6654 & .6531 & .6719 & .6825 & .6646 & .6832 \\
DrEM:S & .6344 & {\color[HTML]{0070C0} \underline{.6690}} & {\color[HTML]{0070C0} \underline{.6554}} & .6725 & {\color[HTML]{0070C0} \underline{.6847}} & {\color[HTML]{FF0000} \textbf{.6673}} & .6828 \\
DrEM   & {\color[HTML]{FF0000} \textbf{.6376}} & {\color[HTML]{FF0000} \textbf{.6742}} & {\color[HTML]{FF0000} \textbf{.6562}} & {\color[HTML]{FF0000} \textbf{.6749}} & {\color[HTML]{FF0000} \textbf{.6873}} & {\color[HTML]{0070C0} \underline{.6661}} & {\color[HTML]{FF0000} \textbf{.6879}} \\
\bottomrule
\end{tabular}%
\end{table}

\noindent\textbf{(1) Overall best performance and cross-backbone consistency.} Our
DrEM achieves the best or comparable performance on both EMER and EASQ.
Notably, the baselines exhibit inconsistent gains across backbones. For example, SSM improves pvtr on EASQ but degrades it on EMER, whereas our DrEM consistently outperforms Base on both.
This suggests that the baselines' effectiveness depends on the specific noise conditions of each backbone, while our DrEM addresses the fundamental issue of pxtr prediction noise and thus generalizes across backbones well as a plug-in module without modifying the backbone structure.

\noindent\textbf{(2) Unstable gains of single-side baselines.}
Neither the supervision-side methods (i.e., SSM, PSL) nor the feature-side methods (i.e., GA, LSPR) consistently outperform Base across all tasks.
For example, PSL improves pcmtr and pftr on EMER but degrades pctr and pvtr, while GA degrades pctr and pcmtr on EMER but improves pvtr and pftr.
This indicates that single-side methods are highly task-dependent, as the same intervention can help on some tasks but hurt on others depending on the specific noise characteristics of each task's pxtr.

\noindent\textbf{(3) Complementary roles of DrEM:S and DrEM:F.}
On the EMER backbone, DrEM:S yields larger gains on sparse tasks (e.g., pcmtr, pftr), while DrEM:F is more effective on dense tasks (e.g., pctr, pvtr).
This aligns with intuition, as sparse tasks have less reliable supervision so supervision-side correction is more valuable, while for dense tasks the supervision is already relatively reliable so feature-side stability becomes more impactful.
On EASQ, the pattern is less pronounced, which is likely because EASQ's questionnaire-signal alignment already reduces the supervision-side noise and thus narrows the gap between the two variants.

\subsection{Ablation Analysis w.r.t Perturbations (RQ2)}
\label{sec:ablation}

To answer RQ2, we conduct ablation experiments on both backbones, progressively removing the supervision-side robust loss and the feature-side consistency regularizer, and study the performance under three perturbation strengths $\alpha\in\{0.5,1.0,2.0\}$.
Figure~\ref{fig:ablation} shows the average GAUC of each variant at varying perturbation strengths, where we have the following observations:

\begin{figure}[t]
  \centering
  \includegraphics[width=1.0\columnwidth]{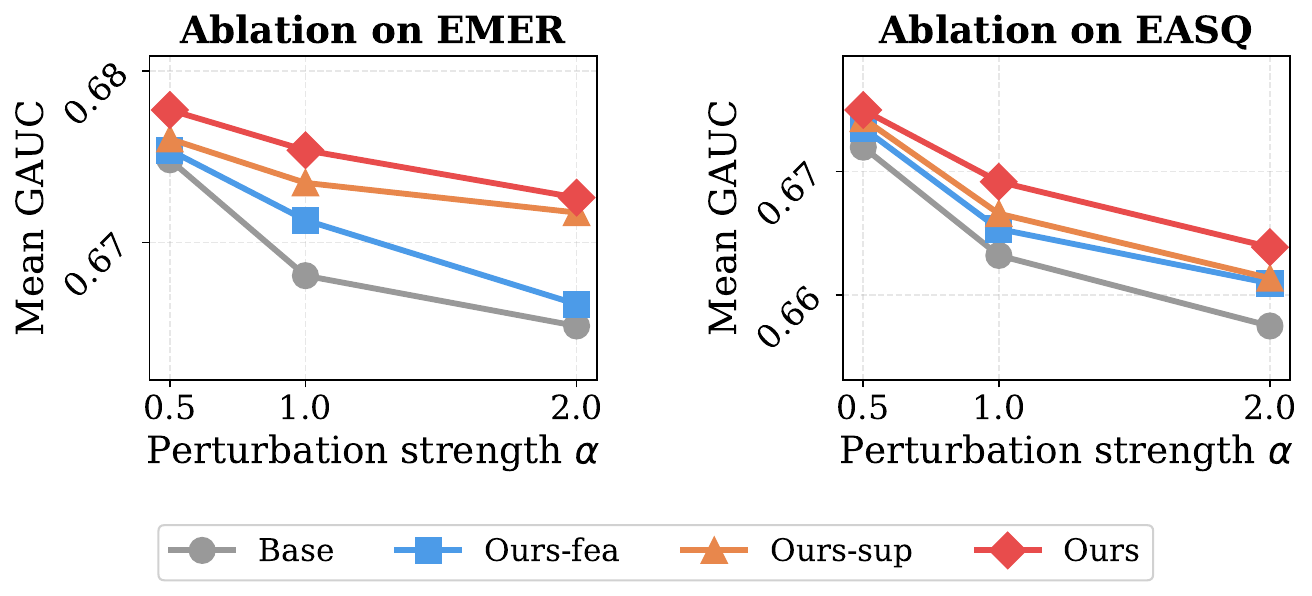}
  \caption{Ablation results. Average GAUC of Base, DrEM:S, DrEM:F, and DrEM under three perturbation strengths. Left: EMER backbone. Right: EASQ backbone.}
  \label{fig:ablation}
\end{figure}

\noindent\textbf{(1) Both sides contribute more as perturbation increases, but with different growth patterns.}
The gain from the supervision-side correction is limited at low perturbation but rises significantly as perturbation grows.
This is consistent with Theorem~\ref{thm:superiority}, as a higher flip probability widens the gap between the basic pairwise loss risk and the clean risk, giving the robust loss more room for recovery.
The feature-side regularizer also shows limited gain at low perturbation but becomes critical at high perturbation, indicating that it protects against input shifts that exceed the model's inherent tolerance.

\noindent\textbf{(2) The full method consistently outperforms each single-side variant, confirming the necessity of dual-side intervention.}
The same pxtr noise corrupts both the supervision signal and the input features, and each side of our DrEM handles one of these two effects, so their gains are additive rather than redundant.
As a result, the combination consistently outperforms each single-side variant across all perturbation levels, and the complementary effect persists even as both sides become effective.

\subsection{Online A/B Test (RQ3)}
\label{sec:online_ab}

To answer RQ3 and validate real-world performance, we conduct online A/B tests on the industrial short-video platform.
We set up two parallel experiment groups to verify the generalization of our DrEM on both representative backbones, applying our DrEM on top of production-grade EMER and EASQ backbone.
The experiment runs for 7 days with 5.1\% of the main traffic randomly split at the user level via hash-based assignment to each group.
The control group uses the online baseline model, and the experiment group uses the same baseline augmented with our DrEM.
We evaluate a comprehensive set of user satisfaction metrics covering long-term retention (LT7), implicit consumption behaviors, and explicit positive interactions.
As shown in Table~\ref{tab:online}, our DrEM achieves consistent positive gains over both baselines, with all improvements reaching statistical significance at $p<0.005$.

\begin{table}[t]
\centering
\caption{Online A/B test results. All $p$-values are below $0.005$.}
\label{tab:online}
\begin{tabular}{ccc}
\toprule
Metric & EMER w/ DrEM & EASQ w/ DrEM \\
\midrule
LT7            & $+0.017\%$ & $+0.020\%$ \\
App Stay Time  & $+0.116\%$ & $+0.124\%$ \\
Video View     & $+0.691\%$ & $+0.117\%$ \\
Long View      & $+0.401\%$ & $+0.043\%$ \\
Like           & $+0.625\%$ & $+0.048\%$ \\
Follow         & $+1.197\%$ & $+0.460\%$ \\
Forward        & $+0.683\%$ & $+0.181\%$ \\
Comment        & $+1.388\%$ & $+0.178\%$ \\
\bottomrule
\end{tabular}
\end{table}

Compared with strong production-grade baselines, our DrEM brings consistent and statistically significant improvements on several core online metrics, with two notable patterns:

\noindent\textbf{(1) Sparse behaviors improve more than dense behaviors.}
The gains on active interaction metrics (Follow $+1.197\%$, Comment $+1.388\%$) far exceed those on passive consumption metrics (Video View $+0.691\%$, Long View $+0.401\%$).
This is consistent with our theoretical analysis, since sparse-behavior tasks (e.g., pcmtr, pftr) have larger prediction variance, leading to higher flip probabilities $\varepsilon_{ij}$ in Eq.~(\ref{eq:risk}) and more severe risk deviation under the basic pairwise loss.
The robust correction in our DrEM therefore has greater room for recovery on these tasks.
The same pattern appears in the offline experiments, where DrEM:S yields larger gains on sparse tasks (e.g., pcmtr, pftr) in Table~\ref{tab:overall}, further confirming the value of flip-probability-driven adaptive correction for high-noise tasks.

\noindent\textbf{(2) The marginal gain over EASQ is smaller than over EMER.}
Except for Video View, all improvements relative to EASQ are smaller than those relative to EMER (e.g., Follow $+0.460\%$ vs.\ $+1.197\%$, Comment $+0.178\%$ vs.\ $+1.388\%$).
This does not indicate that DrEM is less effective on EASQ, since EASQ's questionnaire-signal alignment already partially mitigates the supervision-side noise and thus leaves less room for DrEM's supervision-side correction.

\subsection{Further Analysis (RQ4)}
\label{sec:further_analysis}

In this section, we analyze the intrinsic mechanism of our DrEM from the perspective of the stratified distribution of flip probabilities.
We partition the pairs into ten buckets by their estimated preference flip probabilities with a step size of $0.05$, i.e., $[0, 0.05)$, $[0.05, 0.10)$, $\dots$, $[0.45, 0.50)$, and report the GAUC improvement with DrEM over the Base model for each bucket and each pxtr task.
Figure~\ref{fig:flip_heatmap} visualizes the results as a heatmap, from which we have the following observations:

\begin{figure}[t]
  \centering
  \includegraphics[width=0.96\columnwidth]{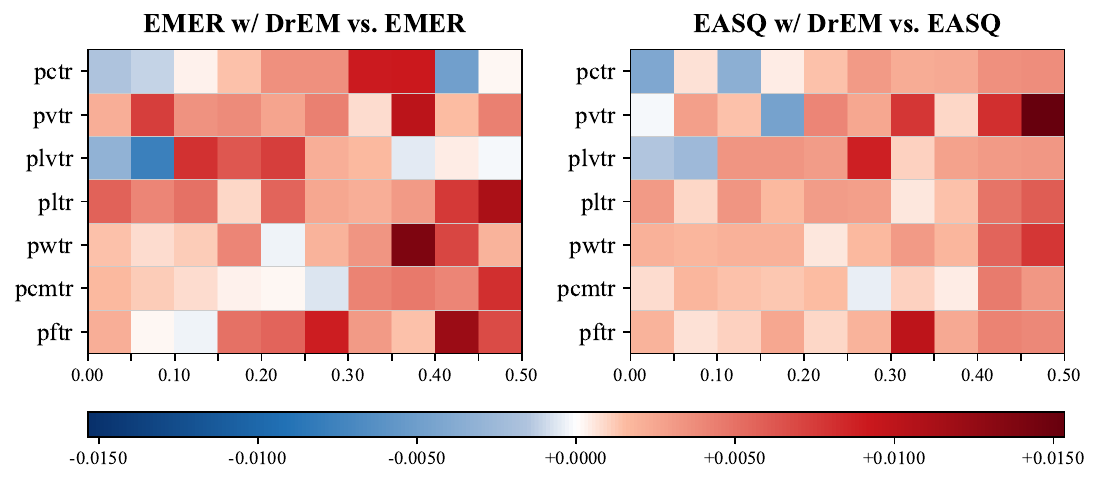}
  \caption{GAUC improvement (Base w/ DrEM vs. Base) stratified by estimated flip probability. Left: EMER backbone. Right: EASQ backbone.}
  \label{fig:flip_heatmap}
\end{figure}

\noindent\textbf{(1) Overall trend, improvement increases with flip probability.}
Across both backbones, the heatmap exhibits a clear left-cold-right-warm pattern, where the GAUC improvement is small or even slightly negative at low flip probabilities and grows as the flip probability increases.
This is consistent with Theorem~\ref{thm:superiority}, as the basic pairwise loss risk deviates from the clean risk by an amount proportional to the flip probability $\varepsilon_{ij}$, so the robust loss has more room for recovery when $\varepsilon_{ij}$ is large.
On the EMER backbone, pctr and plvtr show slightly negative values in the highest flip-probability buckets (e.g., pctr at $[0.40, 0.45)$ and plvtr at $[0.45, 0.50)$).
This is attributable to the extreme scarcity of samples in these buckets, as pairs with such high flip probabilities are exceedingly rare for dense-behavior tasks, and the resulting GAUC estimates are dominated by statistical noise rather than a systematic degradation.

\noindent\textbf{(2) Dense and sparse tasks exhibit different gain profiles.}
For dense-behavior tasks (e.g., pctr, pvtr), the improvement is close to zero or slightly negative at low flip probabilities and rises sharply toward the high end.
For sparse-behavior tasks (e.g., pcmtr, pftr), the improvement is already positive at low flip probabilities and increases more gradually.
This is because sparse-task pxtrs carry higher prediction variance, so even at low flip probabilities the supervision is noisier and the model output is less stable, giving our method more room for correction.
For dense tasks with low variance, the supervision is already reliable at low flip probabilities, leaving little to correct, and the method only becomes beneficial when the flip probability grows large enough to cause noticeable label corruption.
\section{Conclusions}
\label{sec:conclusion}

We propose our DrEM, a dual-side robust ensemble ranking framework that addresses pxtr prediction noise through a shared noise model.
On the supervision side, a risk-denoising robust pairwise loss corrects the empirical risk using estimated pairwise preference flip probabilities. On the feature side, a preference-preserving ranking consistency regularizer stabilizes outputs with item-specific perturbations.
Both components are derived from a shared noise model, aligning dual-side correction targets to the same noise source and synchronizing correction via shared parameters.
Online A/B tests further validate the real-world effectiveness of our DrEM, which brings consistent and statistically significant improvements on core online metrics over production-grade baselines.

\bibliographystyle{ACM-Reference-Format}
\bibliography{references}

\appendix
\appendix

\section{Details of Baselines}
\label{app:baselines}

\noindent\textbf{EMER} \cite{He2025EMER} is an end-to-end framework that models interactions among candidate items via a Transformer and addresses the absence of satisfaction labels through a tailored loss.

\noindent\textbf{EASQ} \cite{Li2026EASQ} aligns item ranking with sparse satisfaction questionnaires through multi-task learning and LoRA while separating questionnaire supervision from dense behavioral signals.

\noindent\textbf{SSM} \cite{Wu2024SSM} links sampled softmax loss to popularity debiasing, hard-negative mining, ranking-metric optimization, and robustness to negative-sample distribution shifts.

\noindent\textbf{PSL} \cite{Yang2024PSL} replaces the exponential function in pairwise softmax loss with alternative activations and is equivalent to distributionally robust BPR under negative-sample distribution shifts.

\noindent\textbf{GaussAug (GA)} \cite{Wang2023Gaussian} generates additional samples from a Gaussian distribution and adversarially selects worst-case augmentations, equivalent to gradient regularization.

\noindent\textbf{LSPR} \cite{Ramazanli2024LSPR} applies Gaussian perturbations to input features and uses perturbed samples as downweighted auxiliary supervision, emphasizing correct predictions rather than output consistency.

\section{Task Details}
\label{app:task_details}

\noindent\textbf{pctr}: predicted click-through rate.

\noindent\textbf{pvtr}: predicted view-through rate.

\noindent\textbf{plvtr}: predicted long-view-through rate.

\noindent\textbf{pltr}: predicted like-through rate.

\noindent\textbf{pwtr}: predicted follow-through rate.

\noindent\textbf{pcmtr}: predicted comment-through rate.

\noindent\textbf{pftr}: predicted forward-through rate.

\end{document}